\documentclass{llncs}

\usepackage[hmargin=1.5in,vmargin=1.5in]{geometry}

\usepackage{xspace}
\usepackage{cite}
\usepackage{amsmath, amssymb}
\usepackage{paralist}
\usepackage{enumerate}
\usepackage{mathtools}

\newcommand{\N}{\ensuremath{\mathbb{N}}}

\newcommand{\eps}{\varepsilon}
\newcommand{\OO}{\ensuremath{\mathcal{O}}}

\newcommand{\EDF}{\ensuremath{\mathsf{EDF}}\xspace}
\newcommand{\LLF}{\ensuremath{\mathsf{LLF}}\xspace}
\newcommand{\OPT}{\ensuremath{\mathsf{OPT}}\xspace}
\newcommand{\Mediumfit}{\ensuremath{\mathsf{MediumFit}}\xspace}
\newcommand{\A}{\ensuremath{\mathsf{A}}\xspace}
\newcommand{\double}{{\sf Double}\xspace}

\usepackage[textsize=tiny,textwidth=2cm,shadow,color=green]{todonotes}

\spnewtheorem{observation}{Observation}{\bfseries}{\itshape}

\title{An $\OO(m^2\log m)$-Competitive Algorithm for Online Machine
  Minimization\thanks{This research was supported by the German
    Science Foundation (DFG) under contract  ME 3825/1. The third
    author was supported by DFG within the research training group `Methods
  for Discrete Structures' (GRK 1408).} }

\author{Lin Chen\inst{1} 
\and Nicole Megow\inst{2} 
\and Kevin Schewior\inst{1} 
}
\institute{Technische Universit\"at Berlin, Institut f\"ur Mathematik,
  Berlin, Germany.\\\email{\{lchen, schewior\}@math.tu-berlin.de}.
\and
Technische Universit\"at M\"unchen, Zentrum f\"ur Mathematik, Garching
bei M\"unchen, Germany. \email{nmegow@ma.tum.de}
}

\begin{document}

\maketitle

\begin{abstract}
We consider the online machine minimization problem in which jobs with
hard deadlines arrive online over time at their release dates. The
task is to determine a feasible schedule on a minimum number of
machines. 
Our main result is a general $\OO(m^2\log m)$-competitive algorithm
for the preemptive online problem, where $m$ is the optimal number of
machines used in an offline solution. This is the first improvement on
an $\OO(\log (p_{\max}/p_{\min}))$-competitive algorithm by Phillips
et al. (STOC~1997), which was to date the best known algorithm even
when~$m=2$. Our algorithm is
$\OO(1)$-competitive for any $m$ that is bounded by a constant.
To develop the algorithm, we investigate two complementary special cases
of the problem, namely, laminar instances and agreeable instances,
for which we provide an $\OO(\log m)$-competitive and an
$\OO(1)$-competitive algorithm, respectively. Our $\OO(1)$-competitive
algorithm for agreeable instances actually produces a non-preemptive
schedule, which is of its own interest as there exists a strong lower
bound of $n$, the number of jobs, for the general non-preemptive online machine minimization problem by Saha (FSTTCS 2013),
which even holds for laminar instances.
\end{abstract}

\section{Introduction}
%
%
We consider the fundamental problem of minimizing the number of
machines that is necessary for feasibly scheduling jobs with
release dates and hard deadlines. We consider the online variant of this problem in which every job
becomes known to the online algorithm only at its release date. We
denote this problem as the {\em online machine minimization problem}. 
We will show that we may restrict to the {\em semi-online} problem variant in which the
online algorithm is given 
the optimal number of machines, $m$, in advance.

Our main result is a  $\OO(m^2\log m)$-competitive algorithm for the preemptive
online machine minimization problem. This is the
first improvement upon a $\OO(\log
\frac{p_{\max}}{p_{\min}})$-competitive algorithm 
\cite{phillipsSTW02}. Our competitive ratio 
depends only on the optimum value, $m$, instead of other input
parameters. In particular, it is constant if the optimum is bounded by a fixed
constant. Such a result was not known even for $m=2$.

\medskip
\noindent\textbf{Previous results.} The preemptive semi-online machine minimization problem, in which the
optimal number of machines is known in advance, has been
investigated extensively by Phillips et al.~\cite{phillipsSTW02}, and there have
hardly been any improvements since then. Phillips et al.\ show a
general lower bound of $\frac{5}{4}$ and leave a huge gap to the
upper bound $\OO(\log \frac{p_{\max}}{p_{\min}})$ on the competitive
ratio for the so-called {\em Least Laxity First}
(\textsf{LLF}) Algorithm. Not so surprisingly, they also rule out that the {\em Earliest Deadline
  First} (\textsf{EDF}) Algorithm may improve on the performance of \textsf{LLF};
indeed they show a lower bound of
$\Omega(\frac{p_{\max}}{p_{\min}})$. It is a wide open question if
preemptive semi-online machine minimization admits a constant
competitive 
ratio
or even a competitive ratio independent of 
the number of jobs or processing times, e.g., an $f(m)$-competitive
algorithm for some function $f$. It is not even known whether a constant competitive algorithm exists for
$m=2$.

The non-preemptive problem is considerably harder than the preemptive
problem. If the set of jobs arrives online over time, then no algorithm
can achieve a competitive ratio sublinear in the number of jobs~\cite{Saha13}. 
However, relevant special cases admit
online algorithms with small constant worst-case guarantees. The
problem with unit processing times was studied in a series of
papers~\cite{KleywegtNST99,ShiY08,KaoCRW12,Saha13,DevanurMPY14} and implicitly in
the context of energy minimization in~\cite{BansalKP07}. It has been shown that an optimal online
algorithm has the exact competitive ratio~$e\approx
2.72$~\cite{BansalKP07,DevanurMPY14}. For non-preemptive scheduling of
jobs with equal deadlines, an upper bound of~$16$ is given
in~\cite{DevanurMPY14}. We are not aware of any previous work on
online machine minimization restricted to agreeable instances. However, in other contexts, e.g., online buffer management~\cite{JezLSS12} and scheduling
with power management~\cite{AlbersMS14,AngelBC14}, it has been studied
as an important and relevant class of instances.

In a closely related problem variant, an online algorithm is given
extra speed to the given number of machines instead of additional unit-speed
machines. The goal is to find an algorithm that requires the minimum extra
speed.  This problem seems much better understood and speedup factors
below~$2$ are known~(see~\cite{phillipsSTW02,lamT99,anandGM11}). However, the power
of speed is much stronger than that of additional machines since it can be viewed to allow
parallel processing of jobs to some extent. None of the algorithms that are known to perform
well for the speed-problem, e.g., \EDF and \LLF, 
are $f(m)$-competitive for any function~$f$ for the
machine minimization problem.

We also mention that the offline problem, in
which all jobs are known in advance, can be solved
optimally in polynomial time if job preemption is allowed~\cite{horn74}. 
Again, the problem complexity increases drastically if
preemption is not allowed. In fact, the problem of  deciding whether
one machine suffices to schedule all the jobs non-preemptively is strongly
NP-complete~\cite{GareyJ77}. It is even open if a constant-factor
approximation exists; a lower bound of~$2-\eps$ was given
in~\cite{CieliebakEHWW04}. The currently best known non-preemptive
algorithm has an approximation factor of $\OO(\sqrt{(\log n)/(\log\log
    n)})$~\cite{ChuzhoyGKN04}. Small constant factors were obtained for special
cases~\cite{CieliebakEHWW04,YuZ09}. However, when slightly increasing the speed
of the machines, then also the general problem can be approximated within a factor~$2$~\cite{ImLMT15}.

\medskip
\noindent\textbf{Our Contribution.} Our main contribution is a new preemptive
online algorithm with a competitive ratio $\OO(m^2\log m)$, where $m$ is
the optimum number of machines. It is the
first improvement upon the $\OO(\log
\frac{p_{\max}}{p_{\min}})$-competitive algorithm by Phillips et
al.~\cite{phillipsSTW02}. Specifically, if the optimum value $m$ is bounded by a constant, our
algorithm is $\OO(1)$-competitive. 

We achieve this algorithm by firstly studying two complementary special cases
of the problem, namely, laminar instances and agreeable instances. In
a laminar instance, if the feasible time intervals for processing (between release date
and deadline) of any two jobs overlap, then one interval is completely
contained in the other. In an agreeable instance, however, the relative order of release dates coincides with that of the corresponding deadline. We provide an 
$\OO(\log m)$-competitive algorithm for laminar instances and an
$\OO(1)$-competitive algorithm for agreeable instances. Then we 
combine both techniques to derive an $\OO(m^2\log m)$-competitive
algorithm for the general problem. 

It is not difficult to see that we may assume that the optimum
  number of machines $m$ is known, i.e, the semi-online model, and that jobs with a small
  processing time relative to the entire time window (``loose'' jobs) are
  easy to schedule. For agreeable instances we show that when scheduling ``tight''
  jobs simply in the middle of their time windows, then there are at most
  $\OO(m)$ jobs running at the same time. 
Our $\OO(1)$-competitive
algorithm 
actually produces a non-preemptive schedule.
This result is of its own interest as there exists a strong lower bound of $n$ for the general non-preemptive (semi)-online machine minimization problem~\cite{Saha13},
which even holds for the special case of laminar instances.

The most difficult special case seems to be the scheduling of ``tight''
  laminar jobs. Here we separate the job-to-machine assignment from
  the scheduling procedure. Upon arrival, we assign a job
  irrevocably to a single machine and simply run \EDF on each machine individually. Thus we restrict ourselves to
  non-migratory schedules, which allows us to keep a better control over
  the remaining processing capacity in a time interval. To assign a job to a machine we consider
  on each machine the previously assigned smallest job that is
  relevant at this point in time, and we find one
  whose laxity is large enough to fully cover the new job's entire
  time window. We then  assign the new job to the same machine. 
  Fitting a job's {\em full} time window (instead of only the processing
  volume) into the laxity of a previously
    assigned job may seem rather restrictive. But jobs are
    ``tight'', i.e., they contribute a significant processing volume
    to the interval, and we gain much more structure for the analysis.
    A greedy assignment turns out to fill up laxities too aggressively. To
    slow down this process we employ a more sophisticated balancing
    scheme in which the laxity of a job is divided into evenly-sized bins and
    jobs are distributed carefully over these bins. Our analysis shows
    that the algorithm is $\OO(\log m)$-competitive.
  
For the general setting it is tempting to use the previous
  algorithms as a black box. However, it is unclear if an online
  partitioning into a small number of agreeable and laminar
  sub-instances is possible. Instead, we propose a more sophisticated variant of
  the assignment procedure for the laminar case, losing an additional factor of
  $m$. Instead of assigning a new job to a single machine, we form
  $\OO(m^2 \log m)$ groups of jobs. Each of these groups is then
  scheduled on $\OO(m)$ machines by generalizing the idea from the
  agreeable setting. 

\smallskip
\noindent\textbf{Outline.} In Section~\ref{sec: pre}, we define the
problem and give first structural insights to the problem. We give an $\OO(1)$-competitive algorithm
for agreeable instances in Section~\ref{sec: agreeable} and an
$\OO(\log m)$-competitive algorithm for laminar instances in
Section~\ref{sec: laminar}. In Section~\ref{sec: general}, we show how
to extend the techniques for the special cases to an $\OO(m^2\log m)$-competitive algorithm for the general problem.

\section{Preliminaries}\label{sec: pre}

\noindent\textbf{Problem Definition.} Given is a set of jobs~$J=\{1,2,\ldots,n\}$ where each job~$j \in J$ has a processing time~$p_j\in \N$, a release date~$r_j\in \N$ which is the earliest possible time at which the job can be processed, and a deadline~$d_j\in \N$ by which it must be completed. The task is to open a minimum number of machines such that there is a feasible schedule in which no job misses its deadline. In a feasible schedule each job~$j\in J$ is scheduled for~$p_j$ units of time within the time window~$[r_j,d_j)$. Each opened machine can process at most one job at any time, and no job is running on multiple machines at the same time. Unless stated differently, we allow job preemption, that is, a job can be preempted at any moment in time and may resume processing later on the same or any other machine. When preemption is not allowed, then a job must run until completion once it has started.

To evaluate the performance of our online algorithms, we perform a {\em competitive analysis} (see e.g.~\cite{borodinEY98}). We call an online algorithm~\A  $c$-{\em competitive} if~$m_{\A}$ machines with~$m_{\A} \leq c\cdot m$ suffice to guarantee a feasible solution for any instance that admits a feasible schedule on~$m$ machines.

\smallskip
\noindent\textbf{Notation.} For any job $j\in J$, the {\em laxity} is
defined as $\ell_j=d_j-r_j-p_j$. We call a job $\alpha$-{\em loose},
for some $\alpha<1$, if $p_j\le \alpha (d_j-r_j)$ and $\alpha$-{\em tight} otherwise. The {\em (processing) interval} of $j$ is $I(j)=[r_j,d_j)$. For a set of jobs $S$, we define $I(S)=\cup_{j\in S}I(j)$. For $I=\cup_{i=1}^k [a_i,b_i)$ where $[a_1,b_1),\dots,[a_k,b_k)$ are pairwise disjoint, we define the \textit{length} of $I$ to be $|I|=\sum_{i=1}^k (b_i-a_i)$.

\smallskip
\noindent\textbf{Characterization of the Optimum.} For $I$ as above,
the \textit{contribution} of a job $j$ to $I$ is
$C(j,I)\coloneqq\max\{0,|I\cap I(j)|-\ell_j\}$, that is, the minimum
processing time that $j$ must receive in $I$ 
in any feasible schedule. The contribution of a job set $S$ to $I$ is
the sum of the individual contributions of jobs in $S$, and we denote
it by $C(S,I)$. Clearly, if $S$ admits a feasible schedule on $m$
machines, $C(S,I)/|I|$ must not exceed $m$. Interestingly, this bound
is tight, 
which we prove in the appendix. 

\begin{theorem}\label{thm: strong-density}
Let $m$ be the minimum number of machines needed to schedule instance $J$ feasibly. Then there exists a union of intervals $I$ with $\lceil C(J,I)/|I|\rceil=m$ but none with $\lceil C(J,I)/|I|\rceil>m$.
\end{theorem}

\smallskip
\noindent\textbf{Reduction to the Semi-Online Problem.} We show that
we may assume that the optimum number of machines~$m$ is known in
advance by
losing at most a factor~$4$ in the competitive ratio. To do so, we
employ the general idea of {\em doubling} an unknown
parameter~\cite{chrobakK06}. More specifically, we open additional
machines whenever the optimum solution has doubled. 

Let $\mathsf{A}_{\rho}(m)$ denote a~$\rho$-competitive algorithm for the
semi-online machine minimization problem given the optimum number of
machines~$m$. Further, denote by $m(t)$ the minimum number of machines needed to feasibly schedule all
jobs released up to time~$t$. Then our algorithm for the online problem is as follows.

\medskip
{\bf Algorithm \double}:
\begin{compactitem}
\item Let~$t_0=\min_{j\in   J}r_j$.
  For $i=1,2,\ldots$ let $t_i=\min\{t\,|\, m(t) > 2m(t_{i-1})\}$.
\item At any time~$t_i$, $i=0,1,\ldots$, open $2\rho
  m(t_i)$ additional machines. All jobs with $r_j\in [t_{i-1},t_i)$
  are scheduled by Algorithm $\mathsf{A}_{\rho}(2m(t_{i-1}))$  on the
  machines opened at time~$t_{i-1}$.
\end{compactitem}
\medskip

Since the time points~$t_0,t_1,\dots$ as well as $m(t_0),m(t_1),\dots$ can be computed online and $\mathsf{A}_{\rho}$ is assumed to be an algorithm for the semi-online problem, this procedure can be executed online. 
Notice 
that \double does not preempt jobs which would not have been preempted by Algorithm~$\mathsf{A}_{\rho}$.

\begin{theorem}\label{thm: black box}
  Given a $\rho$-competitive algorithm for (non-)preemptive
  semi-online machine minimization, \double is $4\rho$-competitive for
  (non-)preemptive online machine minimization. 
\end{theorem}
\vspace*{-1.5ex}
In the rest of the paper we will thus be concerned with the semi-online problem. 

\medskip 
\noindent\textbf{Scheduling Loose Jobs.} \label{subsec: oa-edf-good}
We show that, for any fixed $\alpha<1$, $\alpha$-loose jobs are easy
to handle via a simple greedy algorithm called  {\em Earliest Deadline
  First} ($\EDF$) that schedules at any time~$m'=\rho m$ unfinished
jobs with the smallest deadline. If every job is feasibly scheduled,
it is $\rho$-competitive. 

\begin{theorem}\label{thm: EDF-small}
Let $\alpha\in(0,1)$. \EDF is $1/(1-\alpha)^2$-competitive 
for any instance that consists only of $\alpha$-loose jobs.
\end{theorem}
As we aim for asymptotical competitive ratios in this paper, we
can from now on assume that all jobs are $\alpha$-tight for a fixed
$\alpha\in (0,1)$.

\section{A Constant-Competitive Ratio for Agreeable Instances}\label{sec: agreeable}

Consider an instance~$J$ in which any two jobs $j,j'\in J$ are
\textit{agreeable}, that is, $r_j<r_{j'}$ implies~$d_j\leq d_{j'}$. We
also call $J$ agreeable. For such instances, we derive an
$\OO(1)$-competitive online algorithm. By Theorem~\ref{thm: black box},
we may assume that the optimum value $m$ is known in advance. Using Theorem~\ref{thm: EDF-small}, we can choose some constant
$\alpha\in(0,1)$ and schedule $\alpha$-loose jobs by \EDF on a
separate set of machines. It remains to show an $\OO(1)$-competitive
algorithm for $\alpha$-tight jobs.

We define the \textit{$\beta$-interval} of a job $j$, for some $\beta
\leq 1/2$, as
$I_\beta(j)\coloneqq[r_j+\beta\ell_j,d_j-\beta\ell_j)$. We use the
following simple non-preemptive algorithm.

\medskip

\noindent {\bf Algorithm \Mediumfit}: Each job $j$ is scheduled in the middle of its
feasible time window, i.e., during $I_{1/2}(j)$. If for some job~$j$ there
is no vacant machine during $I_{1/2}(j)$, then we open a new
machine. 

\medskip

For the analysis of \Mediumfit, let two jobs $j,j'$ be
$\beta$-{\em agreeable} if they are agreeable and their $\beta$-intervals
have a non-empty intersection. We give an upper bound on the number of jobs that can be
$\beta$-agreeable with a single job. In fact, we will prove a
  stronger result than actually needed for this section, but it will
  be useful for analyzing the general algorithm in Section~\ref{sec: general}.

\begin{lemma}\label{lem: strong-agreeable}
Consider a job $j$ and some fixed $\alpha\in (0,1)$ and $\beta\in(0,1/2]$. Then there exist at most $\mathcal{O}(m)$ $\alpha$-tight jobs $j'$ with $\ell_{j'}\leq\ell_j$ such that $j$ and $j'$ are $\beta$-agreeable.
\end{lemma}
\begin{proof}
Consider all such $j'$. We estimate their contributions to the interval(s) $$I=[r_j-2\ell_j,r_j+2\ell_j)\cup [d_j-2\ell_j,d_j+2\ell_j),$$ which has a total length of at most $8\ell_j$. There are two possibilities for $j'$.

\begin{enumerate}[\indent{Case }1:]
	\item We have $|I(j')|\ge 2\ell_j$. As $I(j')$ contains $r_j$ or $d_j$, we have $|I\cap I(j')|\ge 2\ell_j$. Given that $\ell_{j'}\leq\ell_j$, $j'$ contributes at least $2\ell_j-\ell_{j'}\ge \ell_j$ to $I$.
	\item We have $|I(j')|<2\ell_j$. As a consequence $I(j')\subseteq I$. Observe that $|I(j')|\geq \beta\ell_j$ holds because $j$ and $j'$ are $\beta$-agreeable. As $j'$ is $\alpha$-tight, its contribution to $I$ is at least $\alpha\beta\ell_j$.
\end{enumerate}

Let $n_1$ and $n_2$ be the number of jobs corresponding to the above two cases, respectively. Then the contribution of the $n_1+n_2$ jobs to $I$ is at least $$(n_1+\alpha\beta n_2)\cdot\ell_j\geq(n_1+n_2)\cdot\alpha\beta\ell_j.$$Using Theorem~\ref{thm: strong-density}, the total contribution is upper bounded by $m|I|\le 8m\ell_j$, implying $n_1+n_2\le 8m/\alpha\beta=\OO(m)$.\qed\end{proof}

Now, consider the \Mediumfit schedule. Any two jobs that are processed
at the same time are $\beta$-agreeable for $\beta=1/2$. By
Lemma~\ref{lem: strong-agreeable}, the number of such jobs---and, thus,
the number of required machines---is~$\OO(m)$. Therefore, we have an
$\OO(1)$-competitive algorithm for agreeable instances. 


Notice also that our final schedule is non-preemptive. \Mediumfit (for $\alpha$-tight jobs) is by definition a non-preemptive algorithm. \EDF
(for $\alpha$-loose jobs) on agreeable instances never preempts jobs that
have already started because the jobs released later have a
larger deadline. Since we compete with a weaker optimum in the
non-preemptive setting, our result carries over to this setting. 
 
\begin{theorem}
	On agreeable instances, there is an
        $\mathcal{O}(1)$-competitive algorithm for preemptive and
        non-preemptive online machine minimization.
\end{theorem}


\section{$\OO(\log m)$-Competitiveness for Laminar Instances}\label{sec: laminar}

In this section, we consider \textit{laminar} instances in which for
any two jobs $j,j'$ with \textit{overlapping} intervals, that is, $|I(j)\cap
I(j')|>0$, holds that either $I(j)\subseteq I(j')$ or $I(j')\supseteq
I(j)$. We prove the following result.


\begin{theorem}\label{thm: laminar}
On laminar instances, there is an $\mathcal{O}(\log m)$-competitive algorithm for online machine minimization.
\end{theorem}

We assume that jobs are indexed from $1$ to $n$ in accordance with
their release dates, that is, $j<j'$ implies $r_j\le r_{j'}$. If
$r_j=r_{j'}$, we assume that $j<j'$ implies $d_j\ge d_{j'}$. We say 
that $j$ \textit{dominates} $j'$ (denoted as $j\succ j'$) if $j<j'$
and $I(j)\supseteq I(j')$. We also say $j'$ {\em is dominated} by $j$
and denote this as $j' \prec j$.

\subsection{Description of the Algorithm}

By Theorems~\ref{thm: black box} and~\ref{thm: EDF-small}, it again suffices to consider semi-online scheduling of $\alpha$-tight jobs for some fixed $\alpha\in(0,1)$.

\medskip

\noindent\textbf{Job Assignment.} We open $m'$ machines and will later
show that we can choose  $m'=\OO(m\log m)$. At every release date, we immediately assign a new job to its machine and we never
  revoke this decision. Thus, we will obtain a non-migratory
  schedule.

The assignment procedure is as follows. We consider an unassigned new
job~$j$. If there is a machine that is ``surely free'' at time $r_j$,
that is, there is no job $j'$ whose time window $I(j')$ contains~$r_j$,
then we assign $j$ to this machine. Otherwise we do the following.

We consider all previously assigned jobs $j'$ with $r_j\in
  I(j')$. Since the instance is laminar, there is on each machine a unique
  $\prec$-{\em minimal} job among them (unless intervals have
  identical time intervals, in which case we pick the one with the smallest index),  that is, a job that is not
  dominating any other job on this machine. Essentially, this is the job that has the
  smallest interval length $I(j')$ on a machine. We call each such job
  \emph{candidate}. Now, consider the laminar chain $c_1(j)
  \prec\dots\prec c_{m'}(j)$ of all candidates at time~$r_j$. Again, this chain exists because the instance is laminar. We call $c_i(j)$ the $i$-th \textit{candidate} of $j$. Notice
that $I(j)\subseteq I(c_i(j))$~for~all~$i$, that is, $j$ is dominated by all its candidates.

We want to assign $j$ to one of the machines and need to ensure that
it can feasibly be scheduled. To that end, we will find a
candidate $c_i(j)$ that has enough capacity (i.e., laxity) in its time window to fit
the {\em full} time interval $I(j)$ and assign $j$ to this candidate $c_i(j)$. Clearly, we have to carefully keep track of the
contribution of previously assigned jobs to $c_i(j)$ which reduces
the remaining capacity (i.e., the laxity) within the candidate's time
window. However,  a greedy assignment of jobs to their $\prec$-minimal
candidates fails, and we employ a more sophisticated balancing
scheme.

We partition the laxity of each candidate's time window into~$m'$
equally-sized bins. To assign $j$, we select the smallest $i$ such
that the $i$-th bin of the $i$'th candidate $c_i(j)$ still has a
remaining capacity of at least $|I(j)|$. 
If there is no such $i$ then the assignment of $j$ \textit{fails}.

To make this assignment more precise and to enable us to analyze the
procedure later,  we introduce some notation. We denote a job~$j$ that
has been assigned to its $i$-th candidate~$c_i(j)$ as the $i$-th
\textit{user} of $c_i(j)$. We denote the set of all $i$-th
users of a job $j'$ by $U_i(j')$, for any~$i$. All $i$-th
users of $j'$ will contribute to the $i$-th bin of $j'$. As stated earlier, to
assign $j$, we select
the smallest $i$ such that the $i$-th bin of $c_i(j)$ still has a
remaining capacity of at least $|I(j)|$, that is, 
\begin{equation}\label{eq: alg-laminar}|I(U_i(c_i(j)))\cup I(j)|\leq\frac{\ell_{c_i(j)}}{m'}.\end{equation}
If we find such an $i$ and assign~$j$, then the capacity of the $i$-th
bin of $c_i(j)$ reduces by $|I(j)|$. Otherwise, the assignment of $j$
\textit{fails}.

\medskip
\noindent\textbf{Scheduling.} Given a job-to-machine assignment, we run \EDF on each machine separately.

\subsection{Analysis of the Algorithm}

We first show that our algorithm obtains a
feasible schedule if no job assignment has failed. Then we give a proof of the fact that the job assignment never fails
on instances that admit a feasible schedule. 

\begin{lemma}\label{lem: laminar-scheduling}
	If the job assignment does not fail, our algorithm schedules feasibly.\end{lemma}
\begin{proof}
Consider some job $j$ and note that, whenever it is preempted at some time $t$, there must be a user $j'$ of $j$ with $t\in I(j')$.
According to Inequality~\eqref{eq: alg-laminar}, $j$ is thus preempted for no longer than $$\sum_{i=1}^{m'} |I(U_i(j))|\leq
m'\cdot\frac{\ell_j}{m'}=\ell_j.$$Therefore, it can be processed for $p_j$ time units.
\qed\end{proof}
%
%
It remains to show that the job assignment never fails for some
$m'=\OO(m\log m)$. The proof idea is as follows: We assume that the
algorithm fails to assign a job. We select a critical job set, the
{\em failure set}, which is the set of jobs that contribute (directly or indirectly) to
  the $m'$-th bin of the $m'$-th candidate of the failing job. We take
into account the direct contribution by users of that bin and the
indirect contribution by jobs that filled lower-indexed bins of lower-indexed candidates and thus indirectly contributed to the filling of
the last possible bin. With this failure set we derive a
contradiction using a load argument. 

Let  $j^\star$ be a job whose assignment fails. We iteratively define
its \textit{failure set} $F\subseteq J$ as the union of
the $m'+1$ different sets $F_0,\dots,F_{m'}$. 
To define the base set $F_0$, we also use auxiliary sets $F^{0}_0,\dots,F^{m'}_0$.  

We initialize $F_0^{m'}\coloneqq \{j^\star\}$. Given $F_0^i$, we
construct $F_{i}$ and $F^{i-1}_0$ in the following way. First, $F_{i}$
is defined to be the set of all $\prec$-maximal $i$-th candidates of
jobs in $F_0^i$. Subsequently, $F_0^{i-1}$ is constructed by incrementing the set of users in $F_0^i$ by the $i$'th users of
  jobs in $F_i$. 
Formally, 
\begin{align*}
	&{F}_{i}\coloneqq M_{\prec}\left(\{c_{i}(j) \mid j\in F_0^{i}\}\right)\\
	\text{and }&F_0^{i-1}\coloneqq F_0^{i}\cup \bigcup_{j\in {F}_{i}}U_{i}(j),
\end{align*}
where $M_{\prec}$ is the operator that chooses the $\prec$-maximal
elements from a set of jobs: $M_{\prec}(S):=\{j\in S\mid \nexists
j':j\prec j'\}.$ After $m'$ such iterations, that is, when
$F_0^0,F_1,\dots,F_{m'}$ have been computed, we set $F_0\coloneqq
M_{\prec}(F_0^0)$. Our failure set is $F_0\cup F_1\cup \ldots
  \cup F_{m'}$.

We show some properties of failure sets which will be useful for the
proofs later. To that end, we define the following notation. For two job sets $S_1$ and $S_2$, we write $S_1\prec S_2$ if, for each $j_1\in S_1$, there is a $j_2\in S_2$ with $j_1\prec j_2$. Moreover, a job set $S$ is said to \textit{$\gamma$-block} a job $j$ if there is a subset $S'\subseteq S$ with $S'\prec \{j\}$ and $|I(S')|\geq\gamma\ell_j$. Finally, $S_1$ $\gamma$-\textit{blocks} $S_2$ if it $\gamma$-blocks every job in $S_2$.

\begin{lemma}\label{lem: fs-bp}
Every failure set $F$ has the following properties.
\begin{enumerate}[(i)]
	\item We have $F_0\prec\dots\prec F_{m'}$.
	\item The sets $F_0,\dots,F_{m'}$ are pairwise disjoint.
	\item For all $i=1,\dots,m'$, $F_0$ $\frac{1}{m'}$-blocks $F_{i}$.
\end{enumerate}
\end{lemma}
\begin{proof}

To see (i), we define $C_i\coloneqq\{c_i(j)\mid j\in F_0^i\}$, for every $1\leq i\leq m'$, and $C_0\coloneqq F^0_0$. We first show $C_{i-1}\prec C_{i}$ for every $1\le i\leq m'$. For $i=1$, this directly follows from the construction. Consider $2\leq i\leq m'$. Let $j$ be an arbitrary job in $C_{i-1}$, which is then an $(i-1)$-th candidate of some job in $F^{i-1}_0$, say, job $j'$. According to the construction, we have $j'\in F^i_0$, or $j'$ is the $i$-th user of some job in $C_i$. In both cases, $C_i$ contains the $i$-th candidate of job $j'$, which dominates job~$j$. Hence, $C_{i-1}\prec C_{i}$. 
Since $F_{i-1}$ and $F_i$ are obtained from $C_{i-1}$ and $C_i$ only by deleting dominated jobs, $F_{i-1}\prec F_i$ follows. 

Consider property (ii) and suppose there is some $j\in F_{i}\cap F_{i'}$ for some $i<i'$. Then, by (i), there is also a job $j'\in F_{i'}$ with $j\prec j'$, contradicting the fact that $F_{i'}$ only contains $\prec$-maximal elements.

Consider property (iii). Let $1\leq i\leq m'$ and $j'\in F_i$ be an
arbitrary job. We show that there is some $F'_0\subseteq F_0$ with
$F'_0\prec \{j'\}$ and $|I(F'_0)|\geq \ell_{j'}/m'$. By the construction
of $F_i$, we have $j'=c_i(j)$ for some $j\in F_0^{i}$. Recall the
construction of $F_0^i$: We have $j=j^\star$, or $j$ is the $i'$-th user of some job where $i'\ge i+1$. In both cases, as our algorithm assigns $j$ to the machine of the smallest possible candidate, $j$ could not be assigned to the same
machine as $j'$, which is its $i$-th candidate. Thus, when our
algorithm assigned job $j$, the capacity of the $i$-th bin of $j'$ did not have enough remaining capacity for $j$ to become its user. Hence, $$ |I(U_i(j')\cup\{j\})|>\ell_{j'}/m'.$$ Note that $U_i(j')\cup\{j\}\subseteq F_0^0$ holds and that $U_i(j')\cup\{j\}$ consists of pairwise disjoint jobs. As $F_0$ is obtained by only deleting jobs in $F_0^0$ that are dominated by others, we can always find $F_0'\subseteq F_0$ with $|I(F'_0)|\geq \ell_{j'}/m'$.
\qed\end{proof}

Using a load argument, we can even derive stronger blocking relations.

\begin{lemma}\label{lem: exp-inc}
For any fixed $\alpha\in (0,1)$, there exists $k=\OO(m)$ with the following property. If $F_i$ $\gamma$-blocks $F_{h}$, then $F_{i+k}$ $2\gamma$-blocks $F_{h}$, for all $h>i+k$.
\end{lemma}
\begin{proof}
Consider some arbitrary $j\in F_{h}$. For every
$i'\in\{i,\dots,i+k\}$, we denote by $F_{i'}^\star$ the subset of
$F_{i'}$ consisting of jobs that are dominated by $j$. We prove that,
for $k=\OO(m)$ the inequality $|I(F^\star_{i+k})|\geq 2|I(F^\star_i)|$
holds. By Lemma~\ref{lem: fs-bp} (i) and (ii), we have the chain
$I(F^\star_{i})\subseteq\dots\subseteq I(F^\star_{i+k})$. As all of
these jobs are $\alpha$-tight, we can thus lower bound the
processing-time contribution of $F^\star_{i'}$ to $I(F^\star_{i+k})$ by $\alpha|I(F^\star_{i'})|\geq \alpha|I(F^\star_{i})|$, for all $i'\in\{i,\dots,i+k\}$. The contributions of all these jobs sum up to at least $k\alpha|I(F^\star_{i})|$, which must not exceed $m|I(F^\star_{i+k})|$ by Theorem~\ref{thm: strong-density}. Hence, for $k>2m/\alpha=\OO(m)$ we have $|I(F^\star_{i+k})|\geq 2|I(F^\star_i)|$, which implies the claim.
\qed\end{proof}


\begin{lemma}
	There exists $m'=\OO(m\log m)$ such that the assignment never fails.
\end{lemma}
\begin{proof}
We will choose $m'=\OO(m\log m)$ such that we get a contradiction to the existence of the failure set $F$. Using Lemma~\ref{lem: fs-bp} (iii) and Lemma~\ref{lem: exp-inc}, we can choose $\lambda=\OO(m\log m')$ such that $F_\lambda$ $2$-blocks every $F_i$ with $i>\lambda$. Now we lower bound the total contribution of these $F_i$ to $I(F_\lambda)$: For any $j\in F_i$ where $i\geq\lambda$, its contribution to $I(F_\lambda)$ is $$|I(F_{\lambda})\cap I(j)|-\ell_{j}\ge |I(F_{\lambda})\cap I(j)|/2$$ since $|I(F_{\lambda})\cap I(j)|\ge 2\ell_{j}$. Thus we get as total contribution of $F_i$ to $F_\lambda$ at least
\begin{align*}
	\sum_{j\in F_i}\frac{|I(F_{\lambda})\cap I(j)|}{2} &\ge \frac{1}{2}\cdot \Big|I(F_{\lambda})\cap \big(\bigcup_{j\in F_k}I(j)\big)\Big|
          = \frac{|I(F_{\lambda})\cap I(F_i)|}{2}\ge \frac{|I(F_{\lambda})|}{2}\,,
\end{align*}
where the last inequality follows from $F_{\lambda}\prec F_{i}$. Therefore, the total contribution of all $F_i$ for $i>\lambda$ is at least $(m'-\lambda)|I(F_\lambda)|/2$, which is at most $m|I(F_\lambda)|$ by Theorem~\ref{thm: strong-density}. Thus, we must have $m'\leq \lambda+2m=\OO(m\log m')$, and appropriately choosing $m'=\OO(m\log m)$ yields a contradiction.
\qed\end{proof}

As stated earlier, combining this lemma with Lemma~\ref{lem: laminar-scheduling} shows Theorem~\ref{thm: laminar}.

\section{$\OO(m^2\log m)$-Competitiveness in the General Case}\label{sec: general}

Now, we generalize the methods introduced above 
and prove our main result.


\begin{theorem}\label{thm: general}
There is an $\mathcal{O}(m^2\log m)$-competitive algorithm for online machine minimization.
\end{theorem}

It is not clear if an online separation into a reasonable number of laminar and agreeable
sub-instances exists such that we could utilize the previous algorithms as black boxes. Nevertheless, we take over the notion of domination from the laminar case and
slightly extend it using the previously introduced definition of $\delta$-intervals $I_\delta(j)\coloneqq[r_j+\delta\ell_j,d_j-\delta\ell_j)$. We say job $j$ \textit{$\delta$-dominates} job
$j'$ if job $j$ dominates $j'$ and we have $I(j')\cap I_\delta(j)
\neq\emptyset$. The presented algorithm works for any fixed
$\delta\in(0,1/2)$.

\subsection{Description of the Algorithm}
Again, by Theorems~\ref{thm: black box} and~\ref{thm: EDF-small}, it
suffices to consider semi-online scheduling of $\alpha$-tight jobs where $\alpha$ is fixed. We describe an $\mathcal{O}(m^2\log
m)$-competitive algorithm. 
In its structure, our 
algorithm mainly resembles the previous algorithm for laminar
instances and it incorporates 
ideas from the one for agreeable instances. 

\smallskip
\noindent\textbf{Job Assignment.} Instead of assigning each job
immediately to a fixed machines, we form $g$ initially empty \textit{groups} of jobs. We will later choose $g=\OO(m^2\log m)$
and show that each of these groups can be scheduled on $\OO(m)$
machines. Upon its release, each job $j$ is immediately assigned to
one of the groups based on the assignment of the 
previously released $\delta$-dominating jobs.
If there exists a group in which no job $\delta$-dominates $j$, then
$j$ is assigned to such a group. Otherwise, in contrast to the laminar
case, there might be multiple $\prec$-minimal such jobs in each
group. Among these jobs we now pick the one with the {\em minimum laxity}
from each group, and we denote this set of $g$ jobs as $D$. 

Like in the laminar case, we choose $\tau$ \textit{candidates},
where $\tau$ is independent of $j$ and will be specified later. Again,
these jobs will form a chain $c_1(j) \prec \dots \prec c_\tau(j)$,
where we call $c_i(j)$ the $i$-th candidate of $j$, for all
$i$. Recall that this step is straightforward for a laminar
instance as the jobs in $D$  form such a chain already. Here,
however, this chain is constructed in an iterative way, starting with
iteration $0$: In iteration $i$, we define $c_{\tau-i}(j)$ to be a
maximum-laxity job among those jobs in $D$ dominated by all jobs
selected so far. If no such jobs exist any more before we have found
the $\tau$-th candidate, the job assignment fails. 

Analogous to the laminar case, we pick a candidate $c_i(j)$ and assign
$j$ to the same group. As before, $j$ will be called an $i$-th
\textit{user} of $c_i(j)$ after this assignment, and the set of $h$-th
users of some job $j'$ is denoted by $U_h(j')$, for all $h$. We open
$\tau$ equally-sized bins for each job, where we again pack
intervals of its $h$-th users into the $h$-th bin, for all $h$. In
contrast to the laminar case, the sum of the bin sizes of a job is
only a small fraction of its laxity: Here, we choose the smallest
candidate $c_i(j)$ such that \\[-2ex]
$$|I(U_i(c_i(j)))\cup I(j)|\le\frac{\ell_{c_i(j)}}{qm\tau},$$ 
where $q$ is a constant yet to be determined. The job assignment fails
if we do not find such a candidate. Note that, again in contrast to
the laminar case, users of the same job may overlap, that is,
$|I(U_i(c_i(j)))\cup I(j)|$ is only a lower bound on the sum of the
individual interval lengths of jobs in $U_i(c_i(j))\cup\{j\}$. 

\smallskip
\noindent\textbf{Scheduling.} Consider a group and let~$S\subseteq J$
be the jobs assigned to it. We schedule $S$ on $\OO(m)$ machines in
the following way. Similarly to our algorithm for agreeable instances, each
job $j\in S$ is scheduled only within $I_\delta(j)$. More precisely, at any time~$t$ consider all the unfinished jobs $j\in S$ with $t\in I_{\delta}(j)$. Among all these jobs, we schedule exactly the $\prec$-minimal jobs, each one on a separate machine. Note that our algorithm for laminar instances also schedules exactly the $\prec$-minimal jobs on each machine, which coincides with \EDF for laminar instances.

\subsection{Analysis of the Algorithm}

As in the laminar case, we first show that, if the job assignment
never fails, we obtain a feasible schedule for all jobs. We need the following auxiliary lemma.

\begin{lemma}\label{lem: delta-half-agreeable}
	Let $j$ and $j'$ be two jobs that do not dominate one another but both $\delta$-dominate some job $j^\star$. Then $j$ and $j'$ are $(\delta/2)$-agreeable.
\end{lemma}
\begin{proof}
	We show that 
	\begin{equation}\label{eq: middle-point}
t^\star\coloneqq\frac{r_{j^\star}+d_{j^\star}}{2}\in I_{\frac{\delta}{2}}(j).		
	\end{equation}
If $t^\star\in I_{\delta}(j)$ holds, Equation~\eqref{eq: middle-point} directly follows. Otherwise, let $I$ be the largest interval that contains $t^\star$ and is a subset of $I(j^\star)\setminus I_{\delta}(j)$. Note that either $r_{j^\star}$ borders $I$ to the left or $d_{j^\star}$ borders $I$ to the right. W.l.o.g.~assume that $I=[t,d_{j^\star})$ for some $t>r_{j^\star}$. Now we have $t^\star-t< |I|/2\leq \delta\ell_j/2$. As $I_{\delta/2}(j)$ is obtained from $I_{\delta}(j)$ by enlarging it by $\delta\ell_j/2$ in both directions, Equation~\eqref{eq: middle-point} follows.
		
	By the same argument, we can prove $t^\star\in I_{\delta/2}(j')$, and the claim follows.
\qed\end{proof}

Now we can prove the following statement, where we will pick a suitable $q$.

\begin{lemma}\label{lem: general-feas-sched}
If the job assignment does not fail, our algorithm produces a feasible schedule on $\OO(mg)$ machines.
\end{lemma}

\begin{proof}
Consider an arbitrary job $j$, assigned to group $h$. We first prove that $j$ receives a processing time of $p_j$, that is, it is preempted for at most $(1-2\delta)\ell_j$ within $I_\delta(j)$. Consider an arbitrary time $t\in I_\delta(j)$ when $j$ is preempted, and let $S_t$ be those jobs $j'$ from group $h$ with $t\in I_\delta(j')$. According to the scheduling rule, $j$ is not a $\prec$-minimal job in $S_t$, so there exist jobs in $S_t$ which are dominated, and thus $\delta$-dominated, by $j$. Among them, let $j^\star$ be a $\prec$-maximal one, and consider the time when job $j^\star$ is assigned to group $h$. Then job $j$ is one of the $\prec$-minimal jobs that $\delta$-dominate $j^\star$ in group $h$. According to the algorithm, among all the $\prec$-minimal jobs that $\delta$-dominate $j^\star$ in group $h$, the one of the least laxity is picked and added to the chain of $\tau$ candidates of $j^\star$, and $j^\star$ is eventually assigned (as a user) to this candidate. Thus,
job $j^\star$ is either a user of job $j$, or a user of some other job $j'$ with $\ell_{j'}\le \ell_j$ which is also a $\prec$-minimal job that $\delta$-dominates $j^\star$.

With the above argument, we conclude that job $j$ is preempted  
either due to one of its own users or a user of some other job $j'$ from group $h$ such that $\ell_{j'}\le \ell_j$. By our job assignment, this preempts $j$ by at most $\ell_j/(qm)$ or $\ell_{j'}/(qm)\le \ell_j/(qm)$, respectively. We claim that there are at most $\OO(m)$ such $j'$. To see why, notice that according to Lemma~\ref{lem:
   delta-half-agreeable} jobs $j$ and $j'$ are
 $\delta/2$-agreeable. As $\ell_{j'}\le \ell_j$, according to Lemma~\ref{lem: strong-agreeable} the claim follows. 
Thus in total job $j$ is preempted by at most
 $\OO(m)\cdot\ell_{j}/(qm) \le (1-2\delta)\ell_j$, with $q$ being a sufficiently
 large constant.  
		
 It remains to show that, in the same group $h$, we process at most $\OO(m)$ jobs at any time $t$. This is straightforward: Let $S_t$ be as above. Note that we always process the $\prec$-minimal jobs in $S_t$, and any two of them are $\delta$-agreeable. Hence, according to Lemma~\ref{lem: strong-agreeable}, there are at most $\OO(m)$ such jobs. \qed
\end{proof}

We now show that the job assignment never fails. In contrast to the laminar case, we also have to make sure that it never fails at finding a chain of candidates.

\begin{lemma}\label{lem: chain-selection}
We can choose $\tau=\Omega(g/m)$, and our algorithm will never fail at finding a chain of $\tau$ candidates.
\end{lemma}
\begin{proof}
Let $S$ be some set of jobs that all $\delta$-dominate the same job $j^\star$. Let job $j$ be the one with maximum laxity in $S$. Note that is suffices to show that, except for $\OO(m)$ jobs, all the jobs in $S$ are dominated by $j$. Then, every time when we pick a new candidate, we lose $\OO(m)$ jobs, and with some $\tau=\Omega(g/m)$ the assignment procedure does not fail when selecting the chain of candidates.

Consider any job $j'\in S$ which is not dominated by $j$. If job $j'$ does not dominate $j$ either, then due to the fact that jobs $j$ and $j'$ both $\delta$-dominate $j^\star$, $j'$ is $(\delta/2)$-agreeable with $j$ according to Lemma~\ref{lem: delta-half-agreeable}. By applying Lemma~\ref{lem: strong-agreeable}, we get that there are at most $\OO(m)$ such jobs $j'$. Otherwise $j'$ dominates $j$. As we have $\ell_{j'}\leq\ell_j$ and $I(j')\supseteq I(j)$ for each such $j'$, it follows that $j'$ contributes at least $p_j$ to $I(j)$. Using that every job is $\alpha$-tight, Theorem~\ref{thm: strong-density} now allows us to bound the number of these jobs by $m/\alpha=\OO(m)$, too.
\qed\end{proof}

The remainder of the proof again relies on selecting a failure set
whenever a job assignment fails and deriving a contradiction from
that. The arguments are the same as in Section~\ref{sec:
  laminar}. Thus, we omit them here and refer to the appendix.


\bibliographystyle{splncs03}
\bibliography{machinenum}

\newpage
\appendix

\section{Missing Material from Section~\ref{sec: pre}}

\subsection{Proof of Theorem~\ref{thm: strong-density}}

\begin{proof}[Theorem~\ref{thm: strong-density}]
It is easy to see that $C(J,I)\le m|I|$ since the processing volume of $C(J,I)$ must be finished within $I$ of length $|I|$ in any feasible solution, while the optimal solution can finish at most $m|I|$ by using $m$ machines.

It remains to show that there exists some $I^\star$ such that $C(J,I^\star)/|I^\star|>m-1$. 
Notice that the offline scheduling problem can be formulated as a flow problem, and there exists an optimal schedule which makes decisions only at integral points in time. Consider any such schedule that uses $m$ machines.
We denote by $\chi_h$ the number of time slots $[t,t+1)$ during which exactly $h$ machines are occupied. Thus, for any optimal schedule on $m$ machines that makes decisions only at integral points in time, we obtain a vector $\chi=(\chi_m,\cdots, \chi_0)$. We define a lexicographical order $<$ on these vectors: we have $\chi<\chi'$ if and only if there exists some $h$ with $\chi_h<\chi_h'$ and $\chi_i=\chi_i'$, for all $i<h$. We pick the optimal schedule whose corresponding vector is the smallest with respect to $<$.

We now construct a directed graph $G=(V,A)$ based on the schedule we have picked. Let $V=\{v_0,\dots,v_{d_{\max}-1}\}$ where $v_t$ represents the slot $[t,t+1)$.
Let $\phi_t$ be the number of machines occupied during $[t,t+1)$. An arc from $v_i$ to $v_k$ exists iff $\phi_i\ge \phi_k$ and there exists some job $j$ with $[i,i+1)\cup [k,k+1)\subseteq I(j)$ whereas $j$ is processed in $[i,i+1)$ but not in $[k,k+1)$.
Intuitively, an arc from $v_i$ to $v_k$ implies that one unit of the workload in $[i,i+1)$ could be carried to $[k,k+1)$.

We claim that in $G$ there is no (directed) path which starts from some $v_i$ with $\phi_i=m$, and ends at some $v_\ell$ with $\phi_\ell\le m-2$. Suppose there exists such a path, say $(v_{i_1},\dots, v_{i_\ell})$ with $\phi_{i_1}=m$ and $\phi_{i_\ell}\le m-2$. Then we alter the schedule such that we move one unit of the workload from $[i_s,i_s+1)$ to $[i_{s+1},i_{s+1}+1)$, for all $s<\ell$. By doing so, $\phi_{i_1}$ decreased and $\phi_{i_\ell}$ increased by $1$ each. By $\phi_{i_1}=m$ and $\phi_{i_\ell}\le m-2$, we get that $\chi_m$ decreases by $1$, contradicting the choice of the schedule.

Consider $V_m=\{v_i\mid\phi_i=m\}$ and let $P(V_m)=\{v_{i_1},\dots,v_{i_\ell}\}$ be the set of vertices reachable from $V_m$ via a directed path (trivially, $V_m\subseteq P(V_m)$). The above arguments imply that for any $v_i\in P(V_m)$, it holds that $\phi_i\ge m-1$. We claim that $\sum_j C(j,I^\star)=\sum_h \phi_{i_h}$ for $I^\star=\{[i_h,i_h+1)\mid 1\le h\le \ell\}$. Note there is no arc going out from $P(V_m)$ (otherwise the endpoint of the arc is also reachable and would have been included in $P(V_m)$). That is, we cannot move out any workload from $I^\star$, meaning that the contribution of all jobs to $I^\star$ is included in the current workload in $I^\star$, i.e., $\sum_h \phi_{i_h}$. Thus, $$\frac{C(J,I^\star)}{|I^\star|}\ge \frac{\sum_h\phi_{i_h}}{|I^\star|}=\frac{\sum_h\phi_{i_h}}{|P(V_m)|}.$$
Notice that $\phi_{i_h}$ is either $m$ or $m-1$, and among them there is at least one of value $m$, thus the right-hand side is strictly larger than $m-1$, i.e., $C(J,I^\star)/|I^\star|>m-1$.  
\qed\end{proof}

We remark that this result can be proven also using LP duality.

\subsection{Proof of Theorem~\ref{thm: black box}}

\begin{proof}[Theorem~\ref{thm: black box}]
  Let~$t_0,t_1,\ldots,t_k$ denote the times at which \double opens new machines.
  \double schedules the jobs $J_{i} = \{j\,|\, r_j\in
  [t_{i},t_{i+1})\}$, with $i=0,1,\ldots,k$, using
  Algorithm~$A_{\rho}(2m(t_i))$ on $2\rho m(t_{i})$
  machines exclusively. This yields a feasible schedule since an
  optimal solution for $J_{i}$ requires at
  most~$m(t_{i+1}-1)\leq 2m(t_i)$ machines and the
  $\rho$-competitive subroutine~$A_{\rho}(2m(t_i))$ is
  guaranteed to find a feasible solution given a factor $\rho$ times
  more machines than optimal, i.e.,  $2\rho m(t_i)$.

  It remains to compare the number of machines opened by \double with
  the optimal number of machines~$m$, which is at least~$m(t_k)$. By
  construction it holds that $2m(t_i) \leq m(t_{i+1})$,
  which implies by recursive application that
  \begin{equation}\label{apx-eqn: a}
    2m(t_i) \leq \frac{m(t_k)}{2^{k-i-1}}.
  \end{equation}
  The total number of machines opened by \double is
  \begin{align*}
    \sum_{i=0}^k 2\rho m(t_i) \ \leq \  \sum_{i=0}^k \rho
    \frac{1}{2^{k-i-1}} m(t_k) \
    = \  2\rho \sum_{i=0}^k \
    \frac{1}{2^{i}} m(t_k)
    \ \leq \ 4 \rho m,
   \end{align*}
   which concludes the proof.
   \qed\end{proof}


\subsection{Proof of Theorem~\ref{thm: EDF-small}}
Recall that \EDF is the algorithm that schedules $m'=\rho m$ unfinished jobs with the smallest deadline at any time $t$. Specifically, if there are multiple jobs with the same deadline, we assume that \EDF will break the tie arbitrarily. Notice that the release dates, deadlines, and processing times of all jobs are integers. Thus, if \EDF decides to process a job at time $t$, it actually processes the job during $[t,t+1)$, i.e., \EDF does not change its decision at fractional time points. Therefore, throughout this subsection, we only consider algorithms that make decisions at integral time points and denote by~$\mathsf{A}$ any such algorithm.  

We give some notations. We denote by $w_{\mathsf{A}}(t)$ the total processing volume of jobs that~$\mathsf{A}$ assigns to $[t,t+1)$ (or simply to $t$). Specifically, if algorithm $\mathsf{A}$ makes decisions only at integral time points, this is simply the number of assigned jobs. Further, we let $p_j^{\mathsf{A}}(t)$ denote the remaining processing time of job $j$ at time $t$, $\ell_j^{\mathsf{A}}(t)=d_j-t-p_j^{\mathsf{A}}(t)$. If it is clear from context, we may drop the superscript $\mathsf{A}$ and simply write $p_j(t)$ and $\ell_j(t)$. We let $W_{\mathsf{A}}(t)$ denote the total remaining processing time of all unfinished jobs (including those not yet released) at time $t$, i.e., $W_{\mathsf{A}}(t)=\sum_{j\in J} p_j-\sum_{s=0}^{t-1} w_{\mathsf{A}}(s)$. 
We use $\OPT$ to indicate an optimal~schedule.

A job is called active if it is released but not finished. An algorithm for the (semi-)online machine minimization problem is called {\em busy} if at all times $t$, $w_{\mathsf{A}}(t)<m'$ implies that at time $t$ there are exactly $w_{\mathsf{A}}(t)$ active jobs, where $m'$ is the number of machines used by the algorithm.

We prove Theorem~\ref{thm: EDF-small} by establishing a contradiction based on the workload that \EDF when missing deadlines must assign to some interval.
We do so by using the following work load inequality, which holds for arbitrary busy algorithms. Here, $d_{\max}$ denotes the maximum deadline that occurs in the considered instance.

\begin{lemma}\label{apx-lem: small-busy-load}
Let every job be $\alpha$-loose, $\rho\ge 1/(1-\alpha)^2$ and let $\mathsf{A}$ be a busy algorithm for semi-online machine minimization using $\rho m$ machines. Assume that $\ell_j^\mathsf{A}(t^\prime)\geq 0$ holds, for all $t^\prime\leq t$ and $j$. Then $$W_\mathsf{A}(t)\le W_{\OPT}(t)+\frac{\alpha}{1-\alpha}\cdot m\cdot\left(d_{\max}-t\right).$$
\end{lemma}

\begin{proof}
We prove by induction. The base is clear since $W_A(0)=W_{\OPT}(0)$. Now assume the lemma holds for all $t^\prime<t$. There are three possibilities.

\begin{enumerate}[\indent{Case }1:]
	\item We have $w_\mathsf{A}(t)\le \alpha/(1-\alpha)\cdot m$. By the fact that $\mathsf{A}$ is busy, there is an empty machine at time $t$, meaning that there are at most $\alpha/(1-\alpha)\cdot m$ active jobs. Given that $\ell_j^\mathsf{A}(t^\prime)\geq 0$ for all $t^\prime\leq t$, each of the active jobs has a remaining processing time of no more than $d_{\max}-t$, implying that $W_\mathsf{A}(t)$ is bounded by $\alpha/(1-\alpha)\cdot m\cdot(d_{\max}-t)$ plus the total processing time of the jobs that are not released. As the latter volume cannot exceed $W_{\OPT}(t)$, the claim follows.
	\item We have $w_\mathsf{A}(t)\ge  1/(1-\alpha)\cdot m$. According to the induction hypothesis, it holds that $W_\mathsf{A}(t-1)\le W_{\OPT}(t-1)+\alpha/(1-\alpha)\cdot m\cdot (d_{\max}-t+1)$. Using $w_{\OPT}(t)\le m$, we get $W_{\OPT}(t)\ge W_{\OPT}(t-1)-m.$ For Algorithm $\mathsf{A}$, it holds that $W_\mathsf{A}(t)=W_\mathsf{A}(t-1)-w_\mathsf{A}(t-1)\le W_\mathsf{A}(t-1)-1/(1-\alpha)\cdot m.$ Inserting the first inequality into the second one proves the claim.
	\item We have $\alpha/(1-\alpha)\cdot m<w_\mathsf{A}(t)<1/(1-\alpha)\cdot m$. Again by the fact that $\mathsf{A}$ is busy, there is an empty machine at time $t$, i.e., there are less than $1/(1-\alpha)\cdot m$ active jobs. Again distinguish two cases.
	\item[Case 3a:] We have $p_j(t)\le \alpha (d_j-t)$ for all active jobs. Then the total remaining processing time of active jobs is bounded by $\alpha (d_{\max}-t)\cdot 1/(1-\alpha)\cdot m$. Plugging in the total processing time of unreleased jobs, which is bounded by $W_{\OPT}(t)$, we again get the claim.
	\item[Case 3b:] There exists an active job $j$ with $p_j(t)> \alpha (d_j-t)$. Using that $j$ is $\alpha$-loose, i.e., $p_j\le \alpha(d_j-r_j)$, we get that $j$ is not processed for at least $(1-\alpha)(t-r_j)$ time units in the interval $[r_j, t)$. As Algorithm $\mathsf{A}$ is busy, this means that all machines are occupied at these times, yielding
\begin{align*}
W_\mathsf{A}(t)&\le W_\mathsf{A}(r_j)-(1-\alpha)(t-r_j)\cdot \rho m\\
&\le W_\OPT(r_j)+\frac{\alpha}{1-\alpha}\cdot m\cdot(d_{\max}-r_j)-(1-\alpha)(t-r_j)\cdot \rho m\\
&\le W_\OPT(r_j)+\frac{\alpha}{1-\alpha}\cdot m\cdot(d_{\max}-t) - m\cdot(t-r_j),
\end{align*}
where the second inequality follows by the induction hypothesis for $t=r_j$, and the third one follows by $\rho \ge 1/(1-\alpha)^2$. Lastly, the feasibility of the optimal schedule implies $$W_{\OPT}(t)\ge W_{\OPT}(r_j)-m\cdot (t-r_j),$$ which in turn implies the claim by plugging it into the former inequality.\qed
\end{enumerate}
\end{proof}

We now apply this lemma in the proof of the theorem.

\begin{proof}[Theorem~\ref{thm: EDF-small}]
Let $\rho m$ be the number of machines \EDF uses and consider the schedule produced by \EDF from an arbitrary instance $J$ only consisting of $\alpha$-loose jobs. We have to prove that every job is finished before its deadline if $\rho =1/(1-\alpha)^2$. To this end, suppose that \EDF fails. Among the jobs that are missed, let $j^\star$ be one of those with the earliest release date.

First observe that we can transform $J$ into $J^\prime\subseteq J$ such that $j^\star\in J^\prime$, $\max_{j\in J^\prime} d_j = d_{j^\star}$ and \EDF still fails on $J^\prime$. For this purpose, simply define $J^\prime=\left\{j\in J\mid d_j\leq d_{j^\star}\right\}$ and notice that we have only removed jobs with a later deadline than $\max_{j\in J^\prime} d_j$. This implies that every job from $J^\prime$ receives the same processing in the \EDF schedule of $J^\prime$ as in the one of $J$.

We can hence consider $J^\prime$ instead of $J$ from now on. In the following, we establish a contradiction on the workload during the time interval $[r_{j^\star},d_{j^\star})$. The first step towards this is applying Lemma~\ref{apx-lem: small-busy-load} for an upper bound. Also making use of the feasibility of the optimal schedule, we get:
\begin{align*}
W_{\EDF}(r_{j^\star})&\le W_{\OPT}(r_{j^\star})+\frac{\alpha}{1-\alpha}\cdot m\cdot(d_{\max}-r_{j^\star})\\
&\le \frac{1}{1-\alpha}\cdot m\cdot(d_{\max}-r_{j^\star}).
\end{align*}
We can, however, also lower bound this workload. Thereto, note that job $j^\star$ is not processed for at least $(1-\alpha)(d_{j^\star}-r_{j^\star})$ time units where $d_{j^\star}=d_{\max}$, implying that all machines must be occupied by then, i.e., \[W_{\EDF}(r_{j^\star})> (1-\alpha)\cdot \rho m \cdot (d_{\max}-r_{j^\star}).\] If we compare the right-hand sides of the two inequalities, we arrive at a contradiction if $\rho\geq 1/(1-\alpha)^2$, which concludes the proof.
\qed\end{proof}

\section{Missing Material from Section~\ref{sec: general}}

\subsection{Analysis of the Job Assignment}

We show that the algorithm does not fail to assign any job to a group for some $\tau=\OO(m\log m)$. The basic idea is the same as that for laminar instances, that is, we again select out a critical job set if the algorithm fails to assign a job, and we derive a contradiction by a load argument from this set.

In particular, for a fixed job $j^\star$ whose assignment fails, we iteratively define its \textit{failure set} $F\subseteq J$, which is the union of the $\tau+1$ different sets $F_0,\dots,F_{\tau}$. To define $F_0$, we also use auxiliary sets $F^{0}_0,\dots,F^{\tau}_0$. 

We initialize $F_0^{\tau}\coloneqq \{j^\star\}$. Given $F_0^i$, we construct $F_{i}$ and $F^{i-1}_0$ in the following way. First, $F_{i}$ is defined to be the set of all $\prec$-maximal $i$-th candidates of jobs in $F_0^i$. Subsequently, $F_0^{i-1}$ is constructed by adding all the $i$-th users of jobs in $F_{i}$ to $F_0^i$. Formally,
\begin{align*}
	&{F}_{i}\coloneqq M_{\prec}\left(\{c_{i}(j) \mid j\in F_0^{i}\}\right)\\
	\text{and }&F_0^{i-1}\coloneqq F_0^{i}\cup \bigcup_{j\in {F}_{i}}U_{i}(j),
\end{align*}

where $M_{\prec}$ is the operator that picks out the $\prec$-maximal elements from a set of jobs: $M_{\prec}(S):=\{j\in S\mid \nexists j':j\prec j'\}.$ After $\tau$ such iterations, that is, when $F_0^0,F_1,\dots,F_{\tau}$ have been computed, we set $F_0\coloneqq M_{\prec}(F_0^0)$.

Again, we have the following lemma whose proof is exactly the same as that of Lemma~\ref{lem: fs-bp}. Notice that, due to the difference in the bin sizes, here $F_0$ $\gamma$-blocks each $F_i$ where $\gamma=1/(q m\tau)$, which is different from Lemma~\ref{lem: fs-bp}.

\begin{lemma}\label{lem: general-fs-bp}
Every failure set $F$ has the following properties:
\begin{enumerate}[(i)]
	\item We have $F_0\prec\dots\prec F_{\tau}$.
	\item The sets $F_0,\dots,F_{\tau}$ are pairwise disjoint.
	\item For all $i=1,\dots,\tau$, $F_0$ $\gamma$-blocks $F_{i}$ where $\gamma=1/(q m\tau)$.
\end{enumerate}
\end{lemma}

The following lemma is also true via the same proof as Lemma~\ref{lem: exp-inc}.

\begin{lemma}\label{lem: general-exp-inc}
For any fixed $\alpha\in (0,1)$, there exists $k=\OO(m)$ with the following property. If $F_i$ $\gamma$-blocks $F_{h}$, then $F_{i+k}$ $2\gamma$-blocks $F_{h}$, for all $i+k<h$.
\end{lemma}

\begin{lemma}\label{lem: general-dont-ever-fail}
	There exists $g=\OO(m^2\log m)$ such that the job assignment never fails.
\end{lemma}
\begin{proof}
We will choose $\tau=\OO(m\log m)$ such that we get a contradiction to the existence of the failure set $F$. Using Lemma~\ref{lem: general-fs-bp} (iii) and Lemma~\ref{lem: general-exp-inc}, we can choose $\lambda=\OO(m\log (qm\tau))$ such that $F_\lambda$ $2$-blocks every $F_i$ with $i>\lambda$. Now we lower bound the total contribution of these $F_i$ to $I(F_\lambda)$: For some $j\in F_i$ where $i\geq\lambda$, its contribution to $I(F_\lambda)$ is $$|I(F_{\lambda})\cap I(j)|-\ell_{j}\ge \frac{|I(F_{\lambda})\cap I(j)|}{2}$$ since $|I(F_{\lambda})\cap I(j)|\ge 2\ell_{j}$. Thus we get as total contribution of $F_i$ to $F_\lambda$ at least
\begin{align*}
	\sum_{j\in F_i}\frac{|I(F_{\lambda})\cap I(j)|}{2} &\ge \frac{1}{2}\cdot \Big|I(F_{\lambda})\cap \big(\bigcup_{j\in F_k}I(j)\big)\Big|\\
	&= \frac{|I(F_{\lambda})\cap I(F_i)|}{2}\ge \frac{|I(F_{\lambda})|}{2},
\end{align*}
where the last inequality follows from $F_{\lambda}\prec F_{i}$. Therefore, the total contribution of all $F_i$ for $i>\lambda$ is at least $(\tau-\lambda)|I(F_\lambda)|/2$, which is at most $m|I(F_\lambda)|$ by Theorem~\ref{thm: strong-density}. Thus, we must have $\tau \leq \lambda+2m=\OO(m\log (q m \tau))$, and appropriately choosing $\tau=\OO(m\log m)$ yields a contradiction. By Lemma~\ref{lem: chain-selection}, we can choose $g=\OO(m\tau)=\OO(m^2\log m)$.
\qed\end{proof}
We now obtain Theorem~\ref{thm: general} by putting Lemmas~\ref{lem: general-feas-sched} and~\ref{lem: general-dont-ever-fail} together.
\end{document}